\documentclass[runningheads,a4paper]{llncs}

\usepackage{amssymb}
\usepackage{graphicx}

\usepackage{url}
\urldef{\mailsa}\path|{alfred.hofmann, ursula.barth, ingrid.haas, frank.holzwarth,|
\urldef{\mailsb}\path|anna.kramer, leonie.kunz, christine.reiss, nicole.sator,|
\urldef{\mailsc}\path|erika.siebert-cole, peter.strasser, lncs}@springer.com|
\newcommand{\keywords}[1]{\par\addvspace\baselineskip
\noindent\keywordname\enspace\ignorespaces#1}

\usepackage{graphicx}
\usepackage{latexsym}

\usepackage{t1enc}
\usepackage[cp1250]{inputenc}

\usepackage{latexsym}
\usepackage{amssymb}
\usepackage{tabularx}
\usepackage{amsfonts}
\usepackage{listings} 

\usepackage{colortbl}
\usepackage{xcolor}
\usepackage{pst-all} 
\usepackage{pst-poly}
\newpsobject{showgrid}{psgrid}{subgriddiv=1,griddots=10,gridlabels=6pt}
\usepackage{graphicx}
\usepackage{fancybox}
\usepackage{wrapfig}

\graphicspath{{./}{./fig/}}

\usepackage{algorithm}
\usepackage{algpseudocode}

\newcommand{\con}{\wedge} 
\newcommand{\dis}{\vee} 
\newcommand{\alw}{\Box} 
\newcommand{\imp}{\Rightarrow} 
\newcommand{\som}{\Diamond} 

\newcommand\nonofoot[1]{%
   \begingroup
   \renewcommand\thefootnote{}\footnote{\kern-0.2ex#1}%
   \addtocounter{footnote}{-1}%
   \endgroup
}

\begin{document}

\mainmatter  

\title{Behavior recognition and analysis in smart environments for context-aware applications}

\titlerunning{Behavior recognition and analysis in smart environments...}

%
%
 \author{Rados{\l}aw Klimek}
 \institute{AGH University of Science and Technology\\
 Al. Mickiewicza 30, 30-059 Krak\'ow, Poland\\
 \email{rklimek@agh.edu.pl}
}
\authorrunning{R. Klimek}


%
%

\maketitle

\begin{abstract}
Providing accurate/suitable information on behaviors in sma\-rt environments is
a challenging and crucial task in pervasive computing
where context-awareness and pro-activity are of fundamental importance.
Behavioral identifications enable to abstract higher-level concepts that are interesting to applications.
This work proposes the unified logical-based framework to recognize and analyze behavioral specifications understood as
a formal logic language that avoids ambiguity typical for natural languages.
Automatically discovering behaviors from sensory data streams as formal specifications
is of fundamental importance to build seamless human-computer interactions.
Thus, the knowledge about environment behaviors expressed in terms of temporal logic formulas constitutes
a base for the reactive and precise reasoning processes to support trustworthy,
unambiguous and pro-active decisions for applications that are smart and context-aware.

\end{abstract}

\noindent
\begin{tabularx}{1\linewidth}{|X|}
\hline This is a pre-print author's version of the paper:
R.~Klimek: Behavior recognition and analysis in smart environments
for context-aware applications. \emph{Proceedings of the IEEE
International Conference on Systems, Man, and Cybernetics (SMC
2015), October 9--12, 2015, Hong Kong}, pp.\ 1949--1955. IEEE
Computer Society 2015. Available at:
\texttt{DOI:10.1109/SMC.2015.340} or
\texttt{http://ieeexplore.ieee.org/xpl/articleDetails.jsp?arnumber=7379472}\\
\hline
\end{tabularx}

\section{Introduction}
\label{sec:introduction}


Nowadays smart spaces are filled with different sensors and sensor-like equipments.
A \emph{sensor} is a device that detects events or changes from a physical environment,
that is a devise which is sensitive to a physical stimulus.
These sensors might constitute the IoT spaces (\emph{Internet of Things})
in which objects with unique identifiers create their own scenarios and interactions.
On the other hand,
the decisive feature of smart spaces is \emph{context-awareness}
which stands for the capabilities to examine changes in
the environment and to react to these changes adequately.
Important aspects of context might be:
where you are, who you are with,
and what resources are nearby.
In other words, context is ``...any information that can be used to characterize the situation of an entity.
An entity is a person, place, or object that is considered relevant to the interaction between
a user and an application, including the user and applications themselves''~\cite{Dey-Abowd-2000}.
In software engineering context-awareness means sensing and reacting on the environment.
Sensing and context understanding are necessary and of critical importance to pro-active decisions which
should be interpreted into domain-relevant concepts and situations.

Formal logic allows assertions about actions and behaviors using accurate and precise notations,
eliminating ambiguity common to other languages.
``Logic has simple, unambiguous syntax and semantics.
It is thus ideally suited to the task of specifying information systems''~\cite{Chomicki-Saake-1998}
showing the form of an argument to be valid or invalid.
Knowledge about arguments enable achieving clear thinking and relevant arguments.

The contribution of this paper is a novel and unified logical-based framework to deploy
automatic methods for the behavior recognition and its reliable knowledge representation
through the formalism of temporal logic.
It allows
to support reactive analysis of logical satisfiability,
in order to obtain trustworthy decisions for the dynamically changing smart environment.
Decisions of a system are transparent for users/inhabitants and
satisfy the assumption of context-awareness and pro-activity.
It is demonstrated that this logical framework is expressive enough.
It is also demonstrated that on-line logical reasoning is suitable for sensor data streams.
The semantic tableaux method for temporal logic as a reasoning procedure is considered.
The architecture of a software system (see Figure~\ref{fig:system-architecture}) is proposed,
as well as algorithms (see Algorithm~\ref{alg:building-specification} and Algorithm~\ref{alg:truth-trees}) to generate and interpret logical specifications.
The simple yet illustrative examples are provided,
see Formulas~(\ref{for:specification-sensor-data}) and~(\ref{for:specification-example}) for Algorithm~\ref{alg:building-specification} and
the discussed example for Algorithm~\ref{alg:truth-trees} at the end of Section~\ref{sec:specifications},
as well as related motivating examples in Section~\ref{sec:examples}.
To the best of our knowledge,
this paper presents the first formal study of
both the reactive behavior recognition and deductive-oriented analysis
for context-aware applications over sensor networks.
On the other hand,
this research opens some new directions, especially related to implementation and experiments.

There are many works considering behavior analysis in pervasive computing.
A survey for human activity is provided in the fundamental work~\cite{Aggarwal-Ryoo-2011}.
Features, representations, classification models, and datasets are surveyed.
Work is comprehensive and discusses many important aspects of the domain.
This paper refers to single-layered approaches as considered in~\cite{Aggarwal-Ryoo-2011}.
A survey of activity recognition for wearable sensors is provided in work~\cite{Lara-Labrador-2013}.
A taxonomy according to aspects of response time and learning scheme is introduced.
A couple of systems are qualitatively compared due the mentioned aspects, as well as some other ones.
Formal logic approaches, except for the fuzzy logic, are not considered.
Behavior recognition in smart homes is a topic in work~\cite{Chua-etal-2009},
whose approach influenced in some way this paper,
however,
models base on Hidden Markov Models, which constitutes a different approach in comparison to this one.
In work~\cite{Chen-etal-2013} a hierarchical framework for human activity recognition is presented,
however, the framework focuses on video based activity recognition.
The method of rather manual transformation into logical rules,
is done in an off-line manner,
and reasoning based on the resolution is proposed.
The aim is to discover a semantic gap between
the low level (data) and the high level (human understanding).
In work~\cite{Chen-etal-2014} a similar approach is presented but
formalization is based on a adaptation of temporal relations from
the Allen's temporal interval logic,
and the reasoning process is not considered.
Apart from the issue of a hierarchical approach,
these works influence this paper in such a way that
the formalization of the observed (human) activities is made on the basis of formal logic.
This paper follows work~\cite{Klimek-Kotulski-2014-IE-AITAmI} which concerns
on-the-fly modeling logical specifications and observing behaviors of users/inhabitants,
in other words, logical specifications are understood as knowledge about user preferences.
Work~\cite{Dwyer-etal-1999} proposes patterns for a property specification and
is considered in a more detailed way in the following Sections of this paper,
especially when discussing the so called learning-based approach.
Work~\cite{Magherini-etal-2013} discusses possibilities of using temporal logic and
model checking for the recognition of human activities.
This paper is relatively close to the work,
however, a deduction based approach is proposed.
The novel aspects are  unified logical framework,
basing on a purely logical approach,
and deductive-based reasoning processes to obtain pro-active decisions.

\section{Preliminaries}
\label{sec:preliminaries}

A context model that consists of three layers is shown in Figure~\ref{fig:context-aware-systems},
c.f.\ also~\cite{Klimek-Kotulski-2014-IE-AITAmI}.
\begin{figure}[htb]
\centering
\scalebox{1.0}{
\begin{pspicture}(2.0,4.0) 
\rput(1.4,0.5){\rnode{s1}{}}
\rput(1.4,2.0){\rnode{s2}{}}
\ncline[linewidth=1pt]{->}{s1}{s2}\naput[labelsep=-13pt,nrot=:D]{\textsf{tracking}}
\rput(1.4,2.2){\rnode{s3}{}}
\rput(1.4,4.0){\rnode{s4}{}}
\ncline[linewidth=1pt]{->}{s3}{s4}\naput[labelsep=-13pt,nrot=:D]{\textsf{sensing}}
\rput(.6,4.0){\rnode{s5}{}}
\rput(.6,2.2){\rnode{s6}{}}
\ncline[linewidth=1pt]{->}{s5}{s6}\naput[labelsep=-13pt,nrot=:D]{\textsf{reacting}}
\rput(.6,2.0){\rnode{s5}{}}
\rput(.6,0.5){\rnode{s6}{}}
\ncline[linewidth=1pt]{->}{s5}{s6}\naput[labelsep=-13pt,nrot=:D]{\textsf{influencing}}
\end{pspicture}
}
\includegraphics[width=.6\textwidth]{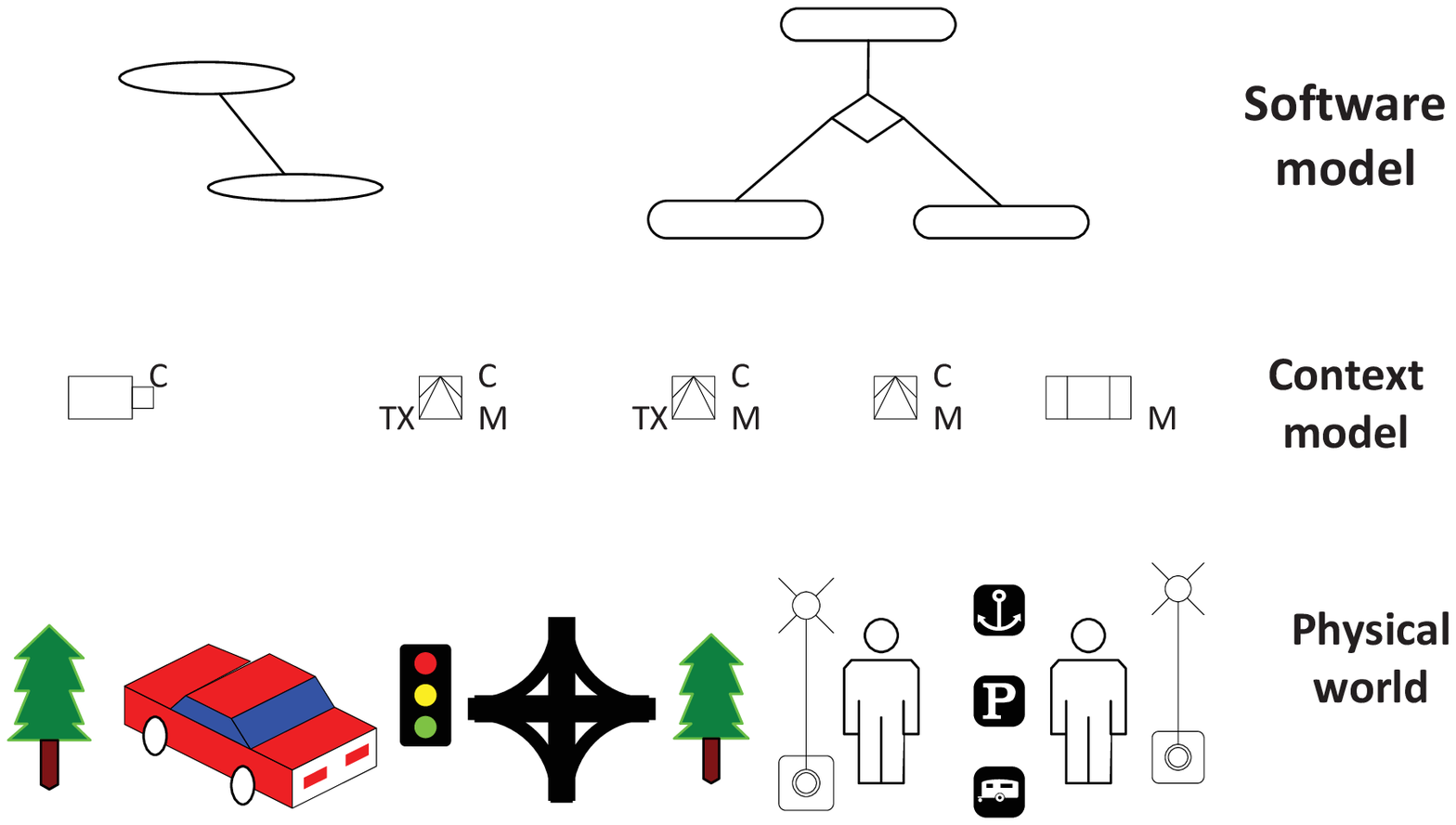}
\caption{A three-layer context model for a smart environment}
\label{fig:context-aware-systems}
\end{figure}
It contains different sensor devices which are distributed in the whole physical area.
It also refers to the concept of \emph{Ambient Intelligence} (AmI),
i.e.\ electronic devices that are sensitive and responsive to the presence of
humans/inhabitants.
Smart applications must both understand context, that is be context-awareness,
and provide pro-activity,
that is act in advance to deal with future situations,
especially negative or difficult ones.
A context-aware system is able to adapt its operations to
the current context without explicit user intervention.

\begin{figure}[htb]
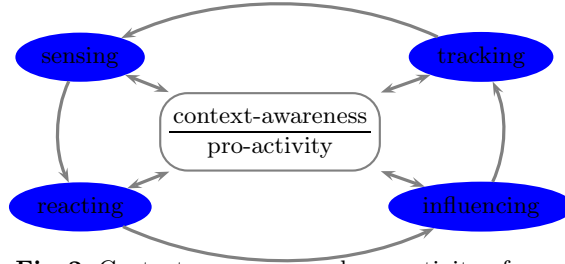

\centering
\scalebox{1}{
{\small
\psmatrix[colsep=.1cm,rowsep=.1cm]
\ovalnode[linecolor=blue,fillstyle=solid,fillcolor=blue]{se}{sensing} &  & \ovalnode[linecolor=blue,fillstyle=solid,fillcolor=blue]{tr}{tracking}\\
 & {\footnotesize \rnode{app}{\psframebox[linecolor=gray,framearc=.7]{\begin{tabular}{c}context-awareness\\ \hline pro-activity\end{tabular}}}} & \\
\ovalnode[linecolor=blue,fillstyle=solid,fillcolor=blue]{re}{reacting} &  & \ovalnode[linecolor=blue,fillstyle=solid,fillcolor=blue]{in}{influencing}
\psset{arcangle=-25,arrows=->,linecolor=gray,linewidth=1.2pt}
\ncarc{tr}{se}\ncarc{se}{re}\ncarc{re}{in}\ncarc{in}{tr}
\psset{arrows=<->,linecolor=gray}
\ncline{app}{se}\ncline{app}{re}\ncline{app}{in}\ncline{app}{tr}
\endpsmatrix
}
}
\caption{Context-awareness and pro-activity of apps}
\label{fig:app}
\end{figure}
The dynamic nature of context models' analysis is shown in Figure~\ref{fig:app},
i.e.\ supplementing Figure~\ref{fig:context-aware-systems},
where different phases are repeated periodically to sense behaviors and
to generate proper system's reactions
enabling context-awareness and pro-activity of applications which operate
in a smart environment.

\emph{Temporal Logic},
and \emph{Propositional Linear Time Temporal Logic} PLTL considered here,
is a branch of formal logic with statements whose valuations depend on
time flows~\cite{Wolter-Wooldridge-2011}.
The reasoning method of \emph{semantic tableaux}
is well known in classical logic but it can be applied in temporal logic~\cite{Hahnle-2001,Klimek-2014-AMCS}.
The method provides \emph{truth trees}.
The \emph{branch} of a tree is a set of nodes/formulas connecting a node with a descendant.
Semantic tableaux is also a \emph{decision procedure} providing,
through \emph{open} branches
(that is, not containing complementary pair/pairs of atomic formulas, e.g.\ $f$ and $\neg f$) and
\emph{closed} branches
(that is, containing complementary pair/pairs of atomic formulas),
the binary answer Yes-No as a result of an inquiry.
\begin{corollary}
\label{th:decision-procedures}
If $F$ is an examined formula and $\Delta$ is a truth tree build for a formula,
then the semantic tableaux method gives answers to the following questions related to the satisfiability problem:
\begin{itemize}
\item formula $F$ is not satisfied iff the finished $\Delta(F)$ is closed;
\item formula $F$ is satisfiable iff the finished $\Delta(F)$ is open;
\item formula $F$ is always  valid iff finished $\Delta(\neg F)$ is closed.
\end{itemize}
\end{corollary}
The proof follows directly from the semantic tableaux method.

\section{Motivating examples}
\label{sec:examples}

Let us consider some examples to illustrate the approach and provide some motivation.
A basic distinction two approaches regarding method of building logical specification is introduced:
\begin{enumerate}
  \item \emph{model-based} -- the case occurs when logical specifications (models) for context-aware systems are prepared in advance;
                                    in other words, the initial specification is not empty,
                                    but new events may affect a particular specification leading to its modification,
                                    it can be used in a decision/reasoning process without any change,
                                    but logical specification can also be dynamically expanded/rebuilt when the system operates,
                                    see the evacuation example below;
  \item \emph{learning-based} -- the case occurs when logical specifications are build on-line, that is in real-time,
                                    during normal operation of a a context-aware system;
                                    in other words, the initial specification is initially empty,
                                    and when new events occur, logical specification is built/rebuilt,
                                    and at any time it can be used in decision-making processes
                                    see work~\cite{Klimek-Kotulski-2014-IE-AITAmI} or the second example below.
\end{enumerate}

The first example discusses an evacuation situation,
i.e.\ people are located inside risk areas (e.g.\ buildings or sport stadiums)
and a dangerous situation occurs.
Context-aware and smart systems should help inhabitants/people by providing trustworthy information about evacuation paths.
The evacuation plan,
expressed as a logical specification $\Sigma$,
and understood as a set of temporal logic formulas,
must be prepared in advance.
(This is a reverse situation comparing other hypothetical cases where logical specifications might be built on-line
i.e.\ when the system operates.)
Formulas describe possible and recommended actions/transitions during the evacuation process.
After the evacuation process has been started and is being carried out,
dynamically changing situations, e.g.\ fire on a passage,
may require extension of $\Sigma$ introducing new formulas describing new situations.
It is done by software agents observing changes in particular areas.
(Graph-based description might contains nodes with different attributes,
such as entrances to corridors or staircases and edges that connect different areas.)

\begin{figure}[htb]
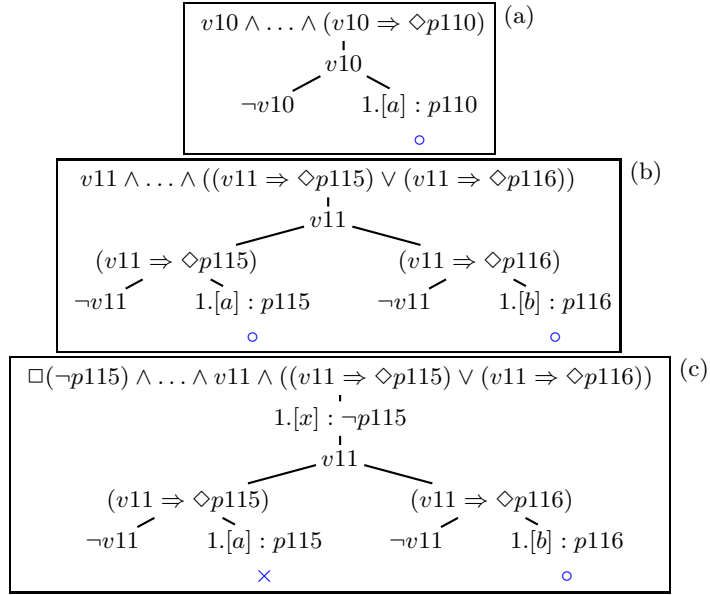

\begin{center}
{\small
\begin{tabular}{c}
\framebox{
\pstree[levelsep=4.0ex,nodesep=2pt,treesep=25pt]
       {\TR{$v10 \con \ldots \con (v10 \imp \som p110)$}}
       {\pstree{\TR{$v10$}}
         {{\TR{$\neg v10$}}{\pstree{\TR{$1.[a]: p110$}}{\textcolor{blue}{$\circ$}}}}}} (a)\vspace{.5mm}\\
\framebox{
\pstree[levelsep=4.0ex,nodesep=2pt,treesep=25pt]
       {\TR{$v11 \con \dots \con ((v11 \imp\som p115) \dis (v11 \imp\som p116))$}}
       {\pstree{\TR{$v11$}}
         {{\pstree{\TR{$(v11 \imp\som p115)$}}{{\TR{$\neg v11$}}{\pstree{\TR{$1.[a]:p115$}}{\textcolor{blue}{$\circ$}}}}}
         {\pstree{\TR{$(v11 \imp\som p116)$}}{{\TR{$\neg v11$}}{\pstree{\TR{$1.[b]:p116$}}{\textcolor{blue}{$\circ$}}}}}}}}
         (b)\vspace{.5mm}\\
\framebox{
\pstree[levelsep=4.0ex,nodesep=2pt,treesep=25pt]
       {\TR{$\alw(\neg p115) \con \ldots \con v11 \con ((v11 \imp\som p115) \dis (v11 \imp\som p116))$}}
       {\pstree{\TR{$1.[x]:\neg p115$}}
       {{\pstree{\TR{$v11$}}
         {{\pstree{\TR{$(v11 \imp\som p115)$}}{{\TR{$\neg v11$}}{\pstree{\TR{$1.[a]:p115$}}{\textcolor{blue}{$\times$}}}}}
         {\pstree{\TR{$(v11 \imp\som p116)$}}{{\TR{$\neg v11$}}{\pstree{\TR{$1.[b]:p116$}}{\textcolor{blue}{$\circ$}}}}}}}}}
         }
         (c)
\end{tabular}
}
\end{center}
\caption{The product of reasoning -- sample truth trees}
\label{fig:truth-trees-reasoning}
\end{figure}
Let $V=\{ \ldots, v10, v11, \ldots, p110, p115, p116, \ldots\}$
are places and passages of a building
for which the evacuation plan is to be prepared.
$\Sigma=\{ \ldots, v10 \imp \som p110, \ldots\}$
is a (small) fragment of the evacuation plan expressed in terms of temporal logic formulas.
When new objects appear in place $v10$ ($v10$ is satisfied), then the reasoning process starts,
see Figure~\ref{fig:truth-trees-reasoning}.a.
The open branch (\textcolor{blue}{$\circ$}) of the tree provides
\emph{literals}, that is atomic formulas or their negations,
$v10$ and $p110$ that satisfy
the initial formula that consists of satisfied $v10$ and conjunction of all formulas that belongs to $\Sigma$.
It allows to identify formula $v10 \imp \som p110$ that describes
the next supporting people action for a particular place, as a part of an evacuation process.
Another situation is shown in Figure~\ref{fig:truth-trees-reasoning}.b.
Let
$\Sigma=\{ \ldots, ((v11 \imp\som p115) \dis (v11 \imp\som p116)), \ldots\}$
is another fragment of an evacuation plan showing the choice of escape routes.
New objects which appear in place $v11$ involve the reasoning process that
provides through two open branches, two subsets of literals $v11$ and $p115$,
and also $v11$ and $p116$.
It means that two different actions are possible,
i.e.\
$v11 \imp \som p115$
or
$v11 \imp \som p116$.
In the last case,
see Figure~\ref{fig:truth-trees-reasoning}.c,
the extension of logical specification $\Sigma$ is discussed.
Supposing that the dynamically changing situation, e.g.\ fire, forces the closure of passage $p115$.
It leads to the need of extending the logical specification by a new formula $\alw(\neg p115)$,
i.e.\ $\Sigma := \Sigma \cup \{\alw(\neg p115)\}$.
Thus, every reasoning process for $p115$ leads to the closed branch (\textcolor{blue}{$\times$}),
i.e.\ the contradiction.
It means that the ``fired'' passage will never be proposed as an action for the evacuation procedure.
This example is also discussed in a more formal way after Algorithm~\ref{alg:truth-trees} in Section~\ref{sec:specifications}.

The above considerations should be supplemented with the following information.
The accepted decomposition procedure in Figure~\ref{fig:truth-trees-reasoning},
as well as labeling, refers to the first-order predicate calculus provided in~\cite{Hahnle-2001}.
In some cases, the outer operator $\alw$ is omitted to simplify considerations/formulas,
in other words, for example, one should write down $\alw(v11 \imp \som p116)$,
however, the well-known rules of generalization/particularization justify the simplified notation.
Reasoning engines have become more available in recent years,
c.f.~\cite{Schmidt-2013-provers},
however,
selection of an appropriate existing prover is not in the scope of this paper.

Another example might refer to the situation when logical specification $\Sigma$,
interpreted as knowledge about user/inhabitant behaviors, is built on-line,
i.e.\ initially $\Sigma=\emptyset$,
and then, observing present users' behaviors,
new temporal logic formulas for particular objects/users are added to set~$\Sigma$.
Work~\cite{Dwyer-etal-1999} discusses methods of obtaining logical specifications from a natural language.
The method is based on pattern recognition.
The consideration in this paper provides a method/idea to obtain logical specifications from
a (technical) language of physical sensors/signals which is less complex when comparing it to a natural language.
Some sample patterns for a ``sensor language'' are provided in Figure~\ref{fig:patterns-ltl}.
\begin{figure}[htb]
\centering
\begin{pspicture}(7,1.5) 
\rput(2,0){\rnode{s1}{}}
\rput(6,0){\rnode{s2}{}}
\psframe[fillstyle=solid,fillcolor=gray,linecolor=gray](2.5,.1)(2.8,-.1)
\psframe[fillstyle=solid,fillcolor=gray,linecolor=gray](4,.1)(4.7,-.1)
\rput(2.65,.25){q}
\rput(4.3,.25){r}
\rput(0.7,0){Response}
\rput(6.9,0){$\alw(q\imp \som r)$}
\ncline[linewidth=1pt]{|-}{s1}{s2}\naput{\textsf{}}
\rput(2,.5){\rnode{s3}{}}
\rput(6,.5){\rnode{s4}{}}
\psframe[fillstyle=solid,fillcolor=gray,linecolor=gray](4.5,.6)(5,.4)
\rput(4.8,.75){r}
\rput(6.5,.5){$\som r$}
\rput(0.7,.5){Existence}
\ncline[linewidth=1pt]{|-}{s3}{s4}\naput{\textsf{}}
\rput(2,1){\rnode{s5}{}}
\rput(6,1){\rnode{s6}{}}
\psframe[fillstyle=solid,fillcolor=gray,linecolor=gray](2,.9)(6,1.1)
\rput(4,1.25){p}
\rput(6.5,1.0){$\alw p$}
\rput(0.75,1){Invariance}
\ncline[linewidth=1pt]{|-}{s5}{s6}\naput{\textsf{}}
\rput(2,1.5){\rnode{s7}{}}
\rput(6,1.5){\rnode{s8}{}}
\rput(6.5,1.5){$\alw \neg p$}
\rput(0.6,1.5){Absence}
\ncline[linewidth=1pt]{|-}{s7}{s8}\naput{\textsf{}}
\end{pspicture}
\caption{Sample PLTL patterns for events $p$, $q$, $r$, etc.}
\label{fig:patterns-ltl}
\end{figure}
Registered (atomic) events for every user might comprise a label for
a physical node (e.g.\ the presence in a node)
and time for the event occurrence (i.e.\ the time stamp),
e.g.\ $\langle p210,t2014.08.14.21.56.00 \rangle$.
These elementary events are translated into logical specifications
when analyzing time of the events and employing (predefined) patterns.
If logical specification $\Sigma$ is built,
then the pro-active decision might be taken when new event, 
say $gt$, occurs and is considered as a kind of trigger.
Triggering is an important aspect for this case.
$\Sigma$, in fact (past) behaviors, is now interpreted as user's preferences to
support a new action of a user.
The entire input formula for the reasoning process might comprise
conjunction of satisfied $gt$ and conjunctions $C(\cdot)$ of the $\Sigma$ formulas,
i.e.\ cumulatively $gt \con C(\Sigma)$.
The reasoning process, and its sub-instances,
might be performed in a similar way as in the previous case
shown in Figure~\ref{fig:truth-trees-reasoning}.

\section{System architecture}
\label{sec:architecture}

The architecture of a proposed system embodied in its well-identified components is
briefly discussed in this Section.
It allows to understand how the system works,
and what are the basic functionalities and services of particular components.

An overall architecture of systems for both model- and learning-based approaches is shown
in Figure~\ref{fig:system-architecture}.
\begin{figure}[htb]
\centering
{\small
\begin{pspicture}(8,7) 
\psset{framearc=0}
\psset{shadow=true,shadowcolor=gray}
\psset{linecolor=black}
\rput(2.3,6.6){\rnode{n:sig1}{signals}}
\rput(2.3,5.6){\rnode{n:sg}{\psframebox
                      {\begin{tabular}{c}
                            \textcolor{black}{\textsc{Signal}}\\
                            \textcolor{black}{\textsc{Manager}}
                      \end{tabular}}}}
\rput(2.3,4.0){\rnode{n:si}{\psframebox
                      {\begin{tabular}{c}
                            \textcolor{black}{\textsc{Signal}}\\
                            \textcolor{black}{\textsc{Interpreter}}
                      \end{tabular}}}}
\rput(0,2.5){\rnode{n:spec}{{$\Sigma$\,\,}}}
\rput(2.3,2.5){\rnode{n:sm}{\psframebox
                      {\begin{tabular}{c}
                            \textcolor{black}{\textsc{Specification}}\\
                            \textcolor{black}{\textsc{Manager}}
                      \end{tabular}}}}
\rput(2.3,1){\rnode{n:su}{\psframebox
                      {\begin{tabular}{c}
                            \textcolor{black}{\textsc{Specification}}\\
                            \textcolor{black}{\textsc{Updater}}
                      \end{tabular}}}}
\rput(6.3,4.0){\rnode{n:re}{\psframebox
                      {\begin{tabular}{c}
                            \textcolor{black}{\textsc{Reasoning}}\\
                            \textcolor{black}{\textsc{Engine}}
                      \end{tabular}}}}
\rput(6.3,2.5){\rnode{n:ri}{\psframebox
                      {\begin{tabular}{c}
                            \textcolor{black}{\textsc{Result}}\\
                            \textcolor{black}{\textsc{Interpreter}}
                      \end{tabular}}}}
\rput(6.3,1){\rnode{n:ap}{\psframebox
                      {\begin{tabular}{c}
                            \textcolor{black}{\textsc{Action}}\\
                            \textcolor{black}{\textsc{Provider}}
                      \end{tabular}}}}
\rput(6.3,0){\rnode{n:sig2}{signals}}

\psset{linewidth=1.1pt}
\psset{shadow=false}
\psset{linecolor=black}
\ncline{->}{n:sig1}{n:sg}
\ncline{->}{n:sg}{n:si}
\ncline[linestyle=dashed]{->}{n:spec}{n:sm}
\ncline[linestyle=dotted]{->}{n:si}{n:sm}
\ncline{->}{n:si}{n:re}\naput{$f$}
\ncline{->}{n:sm}{n:re}\nbput{$\Sigma$}
\ncline{->}{n:re}{n:ri}
\ncline{->}{n:ri}{n:su}
\ncline{->}{n:ri}{n:ap}
\ncline{->}{n:su}{n:sm}
\ncline{->}{n:ap}{n:sig2}
\end{pspicture}
}
\caption{An overall architecture of systems
        (flows: solid lines -- both approaches,
                dashed line -- only model-based approach,
                dotted line -- only learning-based approach)}
\label{fig:system-architecture}
\end{figure}
Signals are gathered
(see tracking/sensing in Figures~\ref{fig:context-aware-systems} and~\ref{fig:app})
from an environment by Signal Manager.
Then signals are interpreted by Signal Interpreter producing temporal logic formulas
generated by an algorithm such as Algorithm~\ref{alg:building-specification},
that is translating events to logical formulas.
If once massive amounts of data are processed (the learning-based approach)
then formulas flow to Specification Manger that stores the basic logical specification $\Sigma$,
that is the current logical model of a system.
If single data is processed
(rather the model-based but also possible in the case of the learning-based approach)
then a formula/formulas are provided to the Reasoning Engine.
The second input for the Reasoning Engine component is logical specification $\Sigma$.
The component performs logical reasoning using the semantic tableaux method,
however, the resolution-based reasoning is also possible.
The output is information, for example, basing on Corollary~\ref{th:decision-procedures},
which is interpreted by Result Interpreter.
It provides two outputs.
The first one allows to update, if necessary, the current logical specification $\Sigma$
(stored in Specification Manager) by Specification Updater through
deleting or adding some new formulas.
The second one allows Action Provider to supply
(see reacting/influencing in Figures~\ref{fig:context-aware-systems} and~\ref{fig:app})
signals to the environment.
Flows in Figure~\ref{fig:system-architecture} are not labeled
(except specification $\Sigma$ and formulas/formula $f$)
since
they would require precise definitions of the flowing data.
On the other hand, their meanings seem intuitive.

Some brief and overall information on methods of Reasoning Engine basing on the semantic tableaux method,
see also Section~\ref{sec:preliminaries} and Corollary~\ref{th:decision-procedures},
is shown in Table~\ref{tab:reasoning}.
\begin{table}[htb]
\centering
\newcolumntype{e}{p{2.0cm}}
\begin{tabularx}{.9\textwidth}{|e|X|}
\hline
Formulations & Remarks\\
\hline
\hline
$C(\Sigma)$, $f \con C(\Sigma)$ & basic logical properties of a specification, open and closed branches, satisfiability, falsification, contradiction\\
\hline
\mbox{$C(\Sigma) \imp f$}, $\neg(C(\Sigma) \imp f)$
& properties that follow, from premises to conclusions, logical consequence, validity, deduction theorem, \emph{reductio ad absurdum}\\
\hline
\end{tabularx}
\caption{Methods of Reasoning Engine}
\label{tab:reasoning}
\end{table}
$C(\Sigma)$ means a conjunction of all formulas constituting logical specification $\Sigma$,
in other words,
a set of formulas $\Sigma$ are interpreted (preprocessed) inside Reasoning Engine as a conjunction of formulas $C(\Sigma)$.
$f$ is a single formula provided by Signal Interpreter.
The reasoning process may comprise many methods and aspects that follow from the input data/formulas,
see formulation in Table~\ref{tab:reasoning},
as well as the assumed reasoning method (truth trees),
for example,
examining satisfiability of the possessed specification,
which happens if a new formula is added to a specification,
whether a property can be inferred from a specification using deductive approach,
etc.

\section{Building and managing specifications}
\label{sec:specifications}

Discovering formal specifications automatically from sensory data streams is discussed below.
The process of building logical specifications should be considered from
a broader point of view which follows from the taxonomy discussed at the beginning of Section~\ref{sec:examples}.

The introduction of a method for building logical specifications,
the physical world, or smart environment, is formally described over a graph structure.
\begin{definition}
An \emph{attributed graph} $G$ is a tuple
$\langle V, E, N, \alpha, S, \beta \rangle$, where
\begin{itemize}
\item $\langle V,E\rangle$ is a directed graph with a set of vertices $V$
and a set of edges or lines $E$,
\item $N$ is a set of labels/names,
\item $\alpha : V \rightarrow N$ is a function that labels vertices,
\item $S$ is a set of labels/sensors, and
\item $\beta : V \rightarrow 2^{S}$ is a function that labels vertices.
\end{itemize}
A \emph{smart environment} $En$ is an attributed graph as defined above.
\end{definition}
$N$ are commonly used (informal) names for vertices, or nodes
(for example: a gate, a crossroad, a staircase, a classroom, etc),
if necessary.
$S$ are sensors located in a node
that detects or measures a physical property and records, indicates, or otherwise responds to it
(for example: tactile sensors,
temperature, humidity and light sensors,
chemical sensors, bio-sensors, etc).
This approach enables the gathering of multiple sensory data in a single node, if necessary.
For example,
on the basis of formal logic,
it can be illustrated by a formula
\begin{eqnarray}
  s_1 \con s_2 \con \neg s_3 \con s_4\label{for:sensors}
\end{eqnarray}
where $s_1, s_2, s_3, s_4 \in S$,
and they are responsible for reading sensory data available in a particular node $v_i \in V$,
say there are the following four data:
temperature exceeded,
humidity exceeded,
high levels of light, and
vibration, respectively.
However,
to simplify the consideration in the rest of the paper
\begin{itemize}
  \item the existence of a single sensor in every node is assumed, and
  \item it is always the object presence sensor/detector that also identifies this object.
\end{itemize}

Let us consider a set of users/inhabitants
$O=\{ o_{1},o_{2},... \}$
that operate in a smart environment.
These users are identified on-line,
i.e.\ when the system operates, and have unique identifiers.
The problem of objects'/users'/inhabitants' unambiguous identification is
a well-known question and it may be done in different ways,
for example by using RFID, PDA devices, biometric data, image scanning, pattern recognition, and others.
The issue of users'/inhabitants' identification is not discussed here.

Events basing on the object presence detection in nodes are registered and
the time-stamp for every event is also registered.
\begin{definition}
\label{def:event}
An \emph{event} $b_i$ is a triple that belongs to $\langle O,V,T \rangle$,
where
\begin{itemize}
\item $O$ is a set of identified users/inhabitants,
\item $V$ is a node of a network, and
\item $T$ is a set of time stamps.
\end{itemize}
A \emph{behavior} $B$ of a smart environment is a set of events $\{ b_{1}, b_{2}, \dots, b_i, \ldots \}$.
\end{definition}
For example,
$b_{i}=\langle idEmily,p0018,t2015.02.11.09.30.15 \rangle$
means that the presence of the $idEmily$ object is observed at
the physical point/area/node $p0018$ of the environment,
and the time stamp assigned to this event is $t2015.02.11.09.30.15$.
Let us note that all nodes that occur in events, or in a behavior,
are equivalent to vertices that occur in an attributed graph,
or a smart environment.
The following notation is introduced.
Let $b_{i}.o_{j}$ be an object $o_{j}$ that belongs to an event $b_i$,
and $b_{i}.v_{k}$ is a node $v_{k}$ that belongs to an event $b_i$, etc.

The algorithm for building logical specifications for every object registered in a smart environment
is given as Algorithm~\ref{alg:building-specification}.
\begin{algorithm}[htb]
\caption{Building logical specifications for objects $O$}
\label{alg:building-specification}
{\normalsize
\begin{algorithmic}[1]
\algrenewcommand\algorithmicrequire{\textbf{Input:}}
\algrenewcommand\algorithmicensure{\textbf{Output:}}
\Require (New) behavior $B$ (non-empty)
\Ensure Logical specifications $L_{i=1,\ldots}$
\State{Divide $B$ into subsets $B_{i=1,\ldots}$ for every object $o_{i=1,\ldots}$}\label{alg:divide}
\For{every $B_i$}

\State $L_i := \emptyset$ \Comment{initiating specification for $o_i$}

\For{$\forall v \in G$}
\If{$v \not\in \{ b_{i}.v_{j} : b_i \in B_i \con j>0 \}$}\label{alg:saf-if}
\State{$L_i := L_i \cup \{ \alw\neg(G.v) \}$}\Comment{saf}\label{alg:saf}
\EndIf{}
\EndFor

\State{Form list $h=[h_1,\ldots,h_n]$ from set $B_{i}$;}\label{alg:form-list}
\State{Sort list $h$ ascending by time stamps;}\label{alg:sort-list}

\State{$l:=1$;}
\Repeat\label{alg:repeat-s}
\State{$k:=l$;}
\While{$(h_{k}.v = h_{l}.v)\con(l<n)$}\label{alg:while-s}
\State{$l:= l+1$;}
\EndWhile\label{alg:while-e}
\If{$h_{k}.v = h_{l}.v$}\label{alg:single-sequence}
\State{$L_i := L_i \cup \{ \som(h_{k}.v) \}$}\Comment{liv1}\label{alg:liv1}
\Else
\If{$h_{k}.v \neq h_{l}.v$}
\State{$L_i := L_i \cup \{ \alw(h_{k}.v \imp \som(h_{l}.v)) \}$}\Comment{liv2}\label{alg:liv2}
\EndIf
\EndIf
\Until{$l=n$}\label{alg:repeat-e}

\EndFor
\end{algorithmic}
}
\end{algorithm}
\emph{Logical specification} $\Sigma$,
or $L_{i}$,
is a set of syntactically correct temporal logic formulas.
The algorithm bases on the analysis of all events that occur in a smart environment.
The algorithm is explained with the remarks given below.
\begin{itemize}
  \item Separate specifications for each object are built (line~\ref{alg:divide});
  \item Every system should be described using both safety and liveness properties~\cite{Alpern-Schneider-1985};
  \item It is tested which nodes are not involved in registered events (line~\ref{alg:saf-if});
  \item The most general form for \emph{safety} (informally: nothing bad will ever happen) is $\alw\neg(p)$,
        i.e.\ some nodes might be never visited (line~\ref{alg:saf}, labeled ``saf'');
        one can consider the \textbf{absence} pattern in terms of Figure~\ref{fig:patterns-ltl};
  \item Auxiliary lists (lines~\ref{alg:form-list} and~\ref{alg:sort-list})
        are created for events that occur for an object;
  \item List $h$ consists of at least one element (line~\ref{alg:form-list});
  \item The repeat loop allows to find all sequences of events following each other
       (lines from~\ref{alg:repeat-s} to~\ref{alg:repeat-e});
  \item The inner loop allows to skip to a different node/event, if any
       (lines from~\ref{alg:while-s} to~\ref{alg:while-e});
  \item The most general form for \emph{liveness} (informally: something good will happen) is $\alw(q \imp \som r)$ or $\som r$,
        i.e.\ some nodes are visited (lines~\ref{alg:liv1} or~\ref{alg:liv2}, labeled ``liv1'' or ``liv2'', respectively);
        one can consider the \textbf{existence} or \textbf{response} patterns in terms of Figure~\ref{fig:patterns-ltl}, respectively;
  \item The existence pattern can occur at most once (line~\ref{alg:single-sequence});
  \item Summing up, temporal logic formulas are produced in three places of the algorithm which are labeled by ``saf'', ``liv1'', and ``liv2''.
\end{itemize}

Let us consider the illustrative example for Algorithm~\ref{alg:building-specification}.
Nodes for a smart environment are $\{ e2, s03, s07, s08, \}$,
i.e.\ three labeled nodes/vertices.
The considered objects/users $O=\{ \ldots,o_{5},\ldots \}$.
A behavior, that is registered events, is
\begin{eqnarray}
B=\{\langle o5,s03,t2015.02.12.09.30.15\rangle,\nonumber\\
\langle o5,s08,t2015.02.12.09.32.40\rangle,\nonumber\\
\langle o5,s08,t2015.02.12.09.33.30\rangle,\nonumber\\
\langle o5,s08,t2015.02.12.09.34.20\rangle,\nonumber\\
\langle o5,s07,t2015.02.12.09.35.20\rangle,\nonumber\\
\langle o5,s07,t2015.02.12.11.37.15\rangle\}\label{for:specification-sensor-data}
\end{eqnarray}
The algorithm produces the following logical specification
\begin{eqnarray}
L_i=\{ \alw\neg(e2),
\alw(s03 \imp \som s08),
\alw(s08 \imp \som s07) \}\label{for:specification-example}
\end{eqnarray}
Every logical specification can be used for the reasoning process as shown
in Figure~\ref{fig:truth-trees-reasoning},
or in Figure~\ref{fig:truth-trees-another} as another example of a truth tree
for Formula~(\ref{for:specification-example}),
where conjunction of all sub-formulas are analyzed.
\begin{figure}[htb]
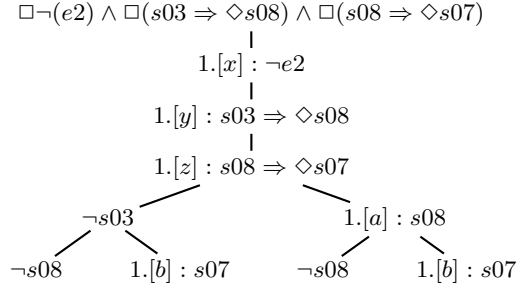

\centering
{\small
\pstree[levelsep=5.0ex,nodesep=2pt,treesep=25pt]
       {\TR{$\alw\neg(e2) \con \alw(s03 \imp \som s08) \con \alw(s08 \imp \som s07)$}}
       {\pstree{\TR{$1.[x]: \neg e2$}}
       {\pstree{\TR{$1.[y]: s03 \imp \som s08$}}
       {\pstree{\TR{$1.[z]: s08 \imp \som s07$}}{\pstree{\TR{$\neg s03$}}{\TR{$\neg s08$}\TR{$1.[b]: s07$}}
                                                 \pstree{\TR{$1.[a]: s08$}}{\TR{$\neg s08$}\TR{$1.[b]: s07$}}}}}}
}
\caption{Another example of a truth tree}
\label{fig:truth-trees-another}
\end{figure}
Many different methods, as well as deductive systems, for truth trees and semantic tableaux are discussed
in work~\cite{Howson-1997} that might help to operate and manipulate efficiently and effectively with truth trees.

The more general remarks for Algorithm~\ref{alg:building-specification}
are given below.
\begin{itemize}
  \item The algorithm produces logical specifications $L_{i}$ for every object that operates in a smart environment;
  \item It should be stressed again that,
        to simplify considerations, the one-sensor case (the object detection) is discussed,
        in other words, Formula~(\ref{for:sensors}) might be replaced by a single atomic sub-formula $s_1$ as an example,
        or, in terms of the algorithm, by $h_{k}.v$ as an example;
  \item The more general issue is the question when the algorithm should operate,
        for example, whenever it is required (on demand) or at ``the end of a day'' (whatever it means),
        this is an open question for future work;
  \item Another open issue is the question of what happens when an ``old'' specification,
        i.e.\ specification obtained as a result of the previous execution, is
        summed, if necessary, with specifications of the current execution,
        then one should examine the entire specification using decision procedures
        mentioned at the end of Section~\ref{sec:preliminaries},
        as Corollary~\ref{th:decision-procedures},
        to discover open and closed branches;
  \item The sketch for the algorithm that unifies, if necessary,
        all specifications obtained from Algorithm~\ref{alg:building-specification}
        is given as Algorithm~\ref{alg:building-environment},
        of course,
        there is no problem to prepare the reverse algorithm,
        that separates logical specifications due to each object.
\end{itemize}
\begin{algorithm}[htb]
\caption{Building logical specification for smart env.\ $En$}
\label{alg:building-environment}
{\normalsize
\begin{algorithmic}[1]
\algrenewcommand\algorithmicrequire{\textbf{Input:}}
\algrenewcommand\algorithmicensure{\textbf{Output:}}
\Require Logical specifications $L_{i=1,\ldots,n}$ (for object $o_{i=1,\ldots,n}$)
\Ensure Logical specification $\Sigma$
\For{every $L_i$}
\For{$\forall f \in L_i$}
\State{attribute formula $f$ uniquely due to object $o_i$}
\EndFor
\EndFor
\State{$\Sigma := \bigcup_{i=1}^{n} L_i$}
\end{algorithmic}
}
\end{algorithm}
\begin{corollary}
\label{cor:algorithm-1}
The following two statements are valid.
\begin{enumerate}
  \item
The time complexity for Algorithm~\ref{alg:building-specification} is expressed by
$\mathcal{O}(o \cdot n)$,
where $o$ is the number of objects that operate in a smart environment,
and $n$ is a number of events registered for each object.
  \item
If a set of all objects and a set of all events are finite,
then Algorithm~\ref{alg:building-specification} always terminates.
\end{enumerate}
\end{corollary}
\begin{proof}
The main, outer loop depends on a number of objects $o$.
The inner, repeat loop depends on a number of events $n$.
Other operations (assignment) and loops (limited number of iterations) give constant costs.
Thus, the time complexity of Algorithm is linearly dependent on the numbers of objects and events.

The number of objects is finite (the for loop),
the number of vertices is limited (the inner for loop),
as well as the number of registered events is limited (the inner repeat loop),
thus, the algorithm always terminates.
\end{proof}

Let us supplement this Section with Algorithm~\ref{alg:truth-trees}
that illustrates more formally considerations following Figure~\ref{fig:truth-trees-reasoning}.
\begin{algorithm}[htb]
\caption{Managing and interpreting truth trees (sketch)}
\label{alg:truth-trees}
{\normalsize
\begin{algorithmic}[1]
\algrenewcommand\algorithmicrequire{\textbf{Input:}}
\algrenewcommand\algorithmicensure{\textbf{Output:}}
\Require Logical specification $\Sigma$; new formula $f$
\Ensure Truth tree $\Delta$; logical specification $\Sigma$; $Open$;
\State{$L:=$ formulas $\Sigma$ that refer to the same object as $f$ refers;}
\State{Build truth tree $\Delta$ for a combined formula $f \con C(L)$;}
\State{$R:=$ select branches of $\Delta$ with literals from formula $f$;}
\State{$Open:=$ select open branches from $R$;}
\State{$Closed:=$ select closed branches from $R$;}
\State{.....}
\State{If necessary, remove/modify formulas from specification~$\Sigma$ basing on literals which belong to $f$ and $Closed$;}\label{alg:remove-s}
\State{$\Sigma:= \Sigma \cup \{f\}$}\Comment{the new basic specification;}\label{alg:remove-e}
\State{.....}
\State{Analyze nodes from $Open$ to provide new actions;}\label{alg:action}
\end{algorithmic}
}
\end{algorithm}
It gives an idea how both Reasoning Engine and Result Interpretation,
shown in Figure~\ref{fig:system-architecture},
work.
The $f \con C(\Sigma)$ case,
see Table~\ref{tab:reasoning},
is taken into account.
It is assumed that initially $\Sigma$ contains no contradiction.
$Closed$ is a set of closed branches of a tree and constitutes a base for further modification of
the basic logical specification $\Sigma$,
if necessary,
removing formulas that contradict with a newly introduced formula.
$Closed'$ is a set of all literals extracted from $Closed$.
$Open$ is a set of open branches of a tree and constitutes a base for selecting satisfiable graph nodes.
$Open'$ is a set of all literals extracted from $Open$.

If necessary,
specification $\Sigma$ is modified,
see lines~\ref{alg:remove-s}--\ref{alg:remove-e},
to remove contradictory formulas from a specification.
This operation is performed using literals which belong to $Closed$/$Closed'$ (contradictory literals)
and $f$ (new formulas, perhaps influencing the basic specification $\Sigma$ through introducing contradictions, if any),
see the example and the last subcase given below.
Analyzing open branches $Open$ to provide actions for a system,
see line~\ref{alg:action},
is a standard procedure,
see the example and all subcases given below.

The illustrative \textbf{example}
to supplement both Algorithm~\ref{alg:truth-trees}
and informal considerations succeeding Figure~\ref{fig:truth-trees-reasoning}
is now provided.
For the (Figure) \ref{fig:truth-trees-reasoning}.a subcase,
$\Sigma=\{ \ldots, v10 \imp \som p110, \ldots\}$ and
$f=v10$.
Then $Open'=\{v10,p110\}$ provides literals that
allow to find the appropriate formula in $\Sigma$,
that is formula $v10 \imp \som p110$.
For the \ref{fig:truth-trees-reasoning}.b subcase,
$\Sigma=\{ \ldots, ((v11 \imp\som p115) \dis (v11 \imp\som p116)), \ldots\}$ and
$f=v11$.
Then $Open'=\{\{v11,p115\},\{v11,p116\}\}$ provides literals leading
to formula $((v11 \imp\som p115) \dis (v11 \imp\som p116))$ describes
formula showing two equivalent movements (passages $p115$ or $p116$).
For the \ref{fig:truth-trees-reasoning}.c subcase,
$f=\alw(\neg p115)$.
Then $Open'=\{v11,p116\}$ and $Closed'=\{\ldots,v11,p115,\dots\}$.
On one hand, $Open'$ allows to point passage $p116$.
On the other hand,
$Closed'$,
showing literals $v11,p115$,
allows to modify a formula as a result of the passage elimination (fire),
that is to replace $((v11 \imp\som p115) \dis (v11 \imp\som p116))$ by
$(v11 \imp\som p116)$.
Then the resulting specification is
$\Sigma = \{\alw(\neg p115), \ldots, v11, (v11 \imp\som p116)\}$.

\textbf{Summing up},
\begin{itemize}
  \item encoding behaviors to logical specifications is a natural process that can be applied to context-aware systems.
  \item There are two different approaches mentioned in the beginning of Section~\ref{sec:examples}.
  \item Some other studies that refer to the implementation and application aspects are open research questions.
        For example,
        the form of a formula located in the root of truth trees,
        that is the disjunction of sub-formulas (the choice between alternatives) or
        conjunction of sub-formulas (satisfiability, contradiction).
        Another example is a method for storing formulas,
        as well as an idea to register multiplicity of formulas/events
        to introduce additional information about the event popularity.
  \item Logical specifications, encoding registered behaviors,
        can be interpreted as preferences understood as a priority in selection.
        Thus, gathering knowledge about preferences is also expressed as logical formulas.
\end{itemize}

\section{Concusion}
\label{sec:conclusion}

This paper presents a method for behavior discovery
as well as the logical satisfiability-oriented reactive analysis
for smart and sensor-based environments to support context-aware and pro-active decisions.
This approach constructs the process for building logical specifications that
fulfill the recognition process providing behavioral specification
in terms of temporal logic formulas.
The proposed unified logical framework is focused on sensor based activity recognition.

Future works should cover more detailed algorithms,
architecture of a multi-agent system and detailed use cases.
Considering graph representations and transformations~\cite{Kotulski-Sedziwy-2011,Kotulski-Sedziwy-2010} is encouraging for
efficient implementation and deploying with presented here logical-oriented approach.
More comparison study with other existing methods and more
theoretical and experimental evaluations are required for future work.

\bibliographystyle{splncs03}
\bibliography{bib-rk,bib-rk-main,bib-rk-pervasive,bib-rk-graph}

\end{document}